\documentclass[11pt]{article}

\usepackage{amsmath,amsthm,amssymb,amsfonts,verbatim}
\usepackage{fullpage}

\usepackage[dvipsnames,usenames]{color}
\usepackage[pagebackref]{hyperref}
\usepackage[colorinlistoftodos]{todonotes}
\usepackage{times}

\providecommand{\abs}[1]{\left|#1\right|}

 \providecommand{\F}{\mathbb{F}}

\date{}

\title{A new class of rank-metric codes and their list decoding beyond the unique decoding radius}

\author{Chaoping Xing \ {\it and} \ Chen Yuan}

\newtheorem{lemma}{Lemma}[section]
\newtheorem{theorem}[lemma]{Theorem}
\newtheorem{prop}[lemma]{Proposition}
\newtheorem{cor}[lemma]{Corollary}

\newtheorem{defn}{Definition}

\theoremstyle{remark}
\newtheorem{rmk}{Remark}

\renewcommand{\epsilon}{\varepsilon}
\renewcommand{\le}{\leqslant}
\renewcommand{\ge}{\geqslant}

\def\period{\Lambda}

\def\proj{\mathrm{proj}}

\newcommand{\mv}[1]{{\mathbf{#1}}}

\newcommand{\vnote}[1]{}

\def\MM{\mathbb{M}}
\def \mC {\mathcal{C}}

\def \mC {\mathcal{C}}

\def \mF {\mathcal{F}}

\def \mP {\mathcal{P}}

\def \Xi {{X^{[i]}}}

\newcommand{\Ga}{\alpha}

\newcommand{\Gg}{\gamma}     
     
\newcommand{\Ge}{\epsilon}

\newcommand{\Gk}{\kappa}
\newcommand{\Gl}{\lambda}    \newcommand{\GL}{\Lambda}

\newcommand{\Gz}{\zeta}
\newcommand{\Gr}{\rho}

\def\rank{{\rm rank}}

\def \bx {{\bf x}}

\def \bz {{\bf z}}

\def \bv {{\bf v}}
\def \bo {{\bf 0}}

\newcommand{\dims}{\kappa}

\begin{document}

\maketitle
\setcounter{page}{0}
\begin{center}  School of Physical \&  Mathematical Sciences, Nanyang Technological University,  Singapore.\\  Emails: {xingcp@ntu.edu.sg; yuan0064@e.ntu.edu.sg}
\end{center}

\begin{abstract}
Compared with classical block codes, efficient list decoding of rank-metric codes  seems more difficult.  The evidences to support this view include:  (i) so far people have not found polynomial time list decoding algorithms of rank-metric codes with decoding radius beyond  $(1-R)/2$ (where $R$ is the rate of code) if ratio of the number of rows over the number of columns is constant, but not very small; (ii) the Johnson bound for rank-metric codes does not exist as opposed to classical codes; (iii) the Gabidulin codes    can not be list decoded beyond half of minimum distance.
Although the list decodability of random rank-metric codes and limits to list decodability have been completely determined, little work on efficient list decoding rank-metric codes  has been done. The only known efficient list decoding of rank-metric codes $\mC$ gives decoding radius up to the Singleton bound $1-R-\Ge$  with positive rate $R$ when $\rho(\mC)$  is extremely small, i.e., $\Theta(\Ge^2)$ , where $\rho(\mC)$ denotes the ratio of the number of rows over the number of columns of $\mC$ \cite[STOC2013]{Guru2013}. It is commonly believed that list decoding of rank-metric codes $\mC$ with not small constant ratio $\rho(\mC)$ is hard.

The main purpose of the present paper is to explicitly construct a class of rank-metric codes $\mC$ with not small constant ratio $\rho(\mC)$
and efficiently list decode these codes with decoding radius beyond $(1-R)/2$. Specifically speaking, let $r$ be a prime power and let $c$ be an integer between $1$ and $r-1$. Let ${\Ge}>0$ be a small real. Let $q=r^{\ell}$ with $\gcd(r-1,\ell n)=1$. Then
there exists an explicit rank-metric code $\mC$ in $\MM_{n\times(r-1)n}(\F_q)$ with rate ${R}$ that is $({\tau}, O(\exp(1/{\Ge^2})))$-list decodable with ${\tau}=\frac c{c+1}\left(1-\frac{r-1}{r-c}\times {R}-{\Ge}\right)$. Furthermore, encoding and list-decoding algorithms are in polynomial time ${\rm poly}(n,\exp(1/{\Ge}))$. The list size can be reduced to $O(1/{\Ge})$ by randomizing  the algorithm. Note that the ratio $\rho(\mC)$ for our code $\mC$ is $1/(r-1)$.
Our key idea is to employ two-variable polynomials $f(x,y)$, where $f$ is linearized in variable $x$ and the variable $y$ is used to ``fold" the code. In other words, rows are used to correct rank errors and columns are used to ``fold" the code to enlarge decoding radius.
Apart from the above algebraic technique, we have to prune down the list. The algebraic idea enables us to pin down the messages
into a structured subspace of dimension linear in the number $n$ of columns. This  ``periodic" structure allows us to pre-encoding the messages to prune down the list. More precisely, we use subspace design introduced in \cite[STOC2013]{Guru2013} to get a deterministic algorithm with a larger constant list size and employ hierarchical subspace-evasive sets introduced in \cite[STOC2012]{Guru2012} to obtain a randomized algorithm with a smaller constant list size.
\end{abstract}

\newpage
\section{Introduction}
Rank-metric codes were first introduced by Delsarte in~\cite{Del} and have found applications in network coding~\cite{KK08} and public-key cryptography~\cite{GPT91,WXS03}. These codes are closely related to space-time codes over finite fields \cite{MV12,Guru2013}. Unique decoding algorithms for rank-metric codes within half minimum distance have been extensively studied \cite{FS,KK08}. However, efficient list decoding of rank-metric codes seems more difficult than that of classical block codes. There are several evidences to support this view. Firstly, people have  not found polynomial-time list decoding algorithms with decoding radius  beyond  $(1-R)/2$ (where $R$ is the rate of code) if ratio of the number of rows over the number of columns is a constant, but not very small. Secondly, the Johnson bound does not exist as opposed to classical codes \cite{RWZ15}. Thirdly, an important class of rank-metric codes introduced by  Gabidulin \cite{Ga85} that are similar to Reed-Solomon codes can not be list decoded beyond half of minimum distance \cite{RWZ15}. The purpose of this paper to design polynomial time list decoding algorithms for rank-metric codes with decoding radius  beyond  $(1-R)/2$.

Before introducing known results and our main results in this paper, we first define list decodability of a  rank-metric code. A rank-metric code over finite filed $\F_q$ is subset of $\MM_{n\times t}(\F_q)$, where $\MM_{n\times t}(\F_q)$ denotes the set of  $n\times t$ matrices over $\F_q$.
Without loss of generality, we always assume $t\ge n$ for a rank-metric code in $\MM_{n\times t}(\F_q)$.
\begin{defn} The rank-metric ball of center $M\in \MM_{n\times t}(\F_q)$ and radius $d$ is defined to be the set $\{X\in \MM_{n\times t}(\F_q):\; {\rm rank}(X-M)\le d\}$.
A rank-metric code $\mC$ is called $(\tau, L)$-list decodable if, for every matrix $M\in \MM_{n\times t}(\F_q)$, there is at most $L$ codewords of $\mC$ in the rank-metric ball of center $M$ of radius  $\tau n$.
\end{defn}

\subsection{Known results}
Unlike list decoding classical codes, there are very few results in literature for efficient list decoding of rank-metric codes. The only known efficient list decoding of rank-metric codes in the asymptotic sense gives decoding radius up to the Singleton bound $1-R-\Ge$ when ratio of the number of rows over the number of columns is $\Theta(\Ge^2)$ \cite[STOC2013]{Guru2013}. On the other hand, list decodability of random rank-metric codes and limits on list decodability of rank-metric codes are completely known \cite{D15,WZ13}. More precisely, we have the following result .

\begin{prop}(see \cite{D15})\label{prop:1.1} Let $n/t$ tend to a fixed constant $\Gr$. Then for any real $R\in (0,1)$, a rank-metric code $\mC\subseteq \MM_{n\times t}(\F_q)$ of rate $R$ that is  $(\tau, L)$-list decodable with $L={\rm poly}(n)$ must obey $R\le (1-\tau)(1-\rho\tau)$. On the other hand, with high probability a random rank-metric code of rate $R$ in $\MM_{n\times t}(\F_q)$ is $(\tau, O(1/\Ge))$-list decodable with $R= (1-\tau)(1-\rho\tau)-\Ge$ for any small real $\Ge>0$. In particular, if  $n/t$ tends to a fixed small constant $\Ge$, then with high probability a random rank-metric code of rate $R$ in $\MM_{n\times t}(\F_q)$ is $(1-R-\Ge, O(1/\Ge))$-list decodable.
\end{prop}
The above result tells that $R= (1-\tau)(1-\rho\tau)$ is the limit to the list decoding of rank-metric codes and moreover most random codes can achieve this limit. The question is how to explicitly construct these codes and efficiently list decode them. It is natural to start with the Gabidulin codes because they are very similar to the classical Reed-Solomon codes. Both of these two classes of codes are constructed from evaluations of polynomials.
As the Reed-Solomon codes can be list decoded up to the Johnson bound \cite{GS99}, people hoped to list decode the Gabidulin codes at least beyond half of the minimum distance, i.e., $\tau>(1-R)/2$. Unfortunately,
it was first shown in \cite{WZ13} that list decodability of the square Gabidulin codes does not exceed the  bound $\tau=1-\sqrt{R}$ and recently it was shown in \cite{RWZ15} that list decodability of the square Gabidulin codes does not exceed half of the minimum distance, i.e., $(1-R)/2$ for a certain family of parameters. This implies that decoding radius of list decoding the square Gabidulin codes is not better than unique decoding.

Inspired by good list decodability of the folded Reed-Solomon codes \cite{GR08}, people started to consider list decoding of folded Gabidulin codes \cite{MV12}. However, the rate of the folded Gabidulin  code in  \cite{MV12} tends to $0$. In 2013, Guruswami and Xing \cite{Guru2013} considered subcodes of the Gabidulin codes via point evaluation in a subfield and showed that list decodability of subcodes of the Gabidulin codes achieves the Singleton bound $\tau=1-R$. However,  the ratio $\Gr=n/t$ of the rank-metric code $\mC\subseteq \MM_{n\times t}(\F_q)$ constructed by Guruswami and Xing \cite{Guru2013}
is $\Theta(\Ge^2)$. This is slightly weaker than random  rank-metric codes where the ratio $\Gr=n/t$ can achieve $\Theta(\Ge)$. So it is still an open problem to explicitly construct rank-metric codes in $ \MM_{n\times t}(\F_q)$ with ratio $\Gr=n/t=\Theta(\Ge)$ and decoding radius $\tau=1-R-\Ge$ and efficiently list decode them.

There has been no much progress on a more interesting case where the ratio $\Gr=n/t$ is not  too small. Hence,
an even more important open problem in the topic of list decoding rank-metric codes is the following
\begin{quote}
{\bf Open Problem.}
For a given constant ratio $\Gr=n/t\in (0,1)$ (not very small), explicitly construct rank-metric codes of rate $R$ in $ \MM_{n\times t}(\F_q)$ with decoding radius $\tau>(1-R)/2$ and efficiently list decode them.
\end{quote}

\subsection{Our results}
The present paper  moves the first step  towards solving the above Open Problem.  We first construct explicit rank-metric codes and  then consider  list decoding  of these rank-metric codes. As a result, we present two decoding algorithms, one deterministic algorithm and one Monte Carlo algorithm. Both the algorithms give the same decoding radius  that is  bigger than $(1-R)/2$. More precisely, we have the followings.
\begin{theorem}{\bf (Main Theorem)}\label{main:1}
Let $r$ be a prime power and let $c$ be an integer between $1$ and $r-1$. Let ${\Ge}>0$ be a small real. Let $q=r^{\ell}$ with $\gcd(r-1,\ell n)=1$.
\begin{itemize}
\item[{\rm (i)}]
There exists an explicit rank-metric code in $\MM_{n\times(r-1)n}(\F_q)$ with rate ${R}$ that is $({\tau}, O(\exp(1/{\Ge^2})))$-list decodable with ${\tau}=\frac c{c+1}\left(1-\frac{r-1}{r-c}\times {R}-{\Ge}\right)$. Furthermore, encoding and list-decoding algorithms are in polynomial time ${\rm poly}(n,\exp(1/{\Ge}))$.
\item[{\rm (ii)}]
With high probability one can randomly sample a rank-metric code in $\MM_{n\times(r-1)n}(\F_q)$ with rate ${{R}}$ that is $({{\tau}}, O(1/{\Ge}))$-list decodable with ${{\tau}}=\frac c{c+1}\left(1-\frac{r-1}{r-c}\times {R}-{\Ge}\right)$. Furthermore, encoding and list-decoding algorithms are in polynomial time ${\rm poly}(n,\exp(1/{\Ge}))$.
\end{itemize}
\end{theorem}

  \begin{rmk}\begin{itemize}
\item[(i)] In the above main theorem, if we fix $r$ and $c$ with $2\le c\le r-1$, then
\[\frac{c}{c+1}\left(1-\frac{r-1}{r-c}\times R\right)>\frac12(1-R)\]
for any $0\le R<\frac{r-c}{r+c}$. This means that our decoding radius  breaks the unique decoding radius for $R\in\left[0,\frac{r-c}{r+c}\right)$. For instance, taking $r=3$ and $c=2$ gives a rank-metric code $\mC\subseteq\MM_{n\times 2n}(\F_q)$ of rate $R$ and decoding radius $\tau=\frac23(1-2R)$ which is bigger than $\frac12(1-R)$ for $R<\frac15$. In this case, the ratio $\rho=n/t$ is $1/2$.
\item[(ii)] By Proposition \ref{prop:1.1}, a rank-metric code $\mC\subseteq\MM_{n\times t}(\F_q)$ of rate $R$ that is $(\tau,L)$-list decodable with $L={\rm poly}(n)$ must obey $R\le (1-\tau)(1-\rho \tau)$, where $\rho$ is the ratio $n/t$. In our case, the ratio $\rho=n/t=1/(r-1)$. Thus, we must have $R\le (1-\tau)\left(1-\frac {\tau}{r-1} \right)$. The decoding radius in the above theorem gives $R\approx \frac{r-c}{r-1}\left(1-\frac{c+1}c\times \tau\right)$ and indeed, one can easily check that
    \[\frac{r-c}{r-1}\left(1-\frac{c+1}c\times \tau\right)<(1-\tau)\left(1-\frac {\tau}{r-1} \right).\]
    \item[(iii)] Unfortunately, our main theorem does not improve the unique decoding bound for square rank-metric codes.  To get square matrices, $r$ has to be $2$. In this case, we can only take $c=1$. Then the decoding radius in the above main theorem gives $\tau=\frac12(1-R)$ which is the same as the unique decoding radius.
 \end{itemize}
 \end{rmk}
In the above theorem, setting $r=\Theta\left(\frac1{\Ge^2}\right)$ and $c=\Theta\left(\frac1{\Ge}\right)$ gives the following corollary.

\begin{cor}\label{cor:1.3}
Let ${\Ge}>0$ be a small real. Let $r=\Theta\left(\frac1{\Ge^2}\right)$ and  $q=r^{\ell}$ with $\gcd(r-1,\ell n)=1$.
\begin{itemize}
\item[{\rm (i)}]
There exists an explicit rank-metric code in $\MM_{n\times(r-1)n}(\F_q)$ with rate ${R}$ that is $({\tau},(1/{\Ge})^{O(\exp(1/{\Ge}^4))})$-list decodable with ${\tau}=1-R-{\Ge}$. Furthermore, encoding and list-decoding algorithms are in polynomial time ${\rm poly}(n,\exp(1/{\Ge}))$.
\item[{\rm (ii)}]
With high probability one can randomly sample a rank-metric code in $\MM_{n\times(r-1)n}(\F_q)$ with rate ${{R}}$ that is $({{\tau}}, O((1/{\Ge})))$.  Furthermore, encoding and list-decoding algorithms are in polynomial time ${\rm poly}(n,(\exp(1/{\Ge}))$.
\end{itemize}
\end{cor}
  \begin{rmk} \begin{itemize}
  \item [(i)] See Remarks \ref{rmk:5} and \ref{rmk:6} for discussion of the list sizes in Corollary \ref{cor:1.3}.
  \item [(ii)]The ratio in the above corollary is $\rho=n/t=1/(r-1)=\Theta(\Ge^2)$. This ratio is the same as the one in  \cite[STOC2013]{Guru2013}. Thus, the above corollary matches the result of \cite[STOC2013]{Guru2013}.
\end{itemize}
 \end{rmk}

\subsection{Our techniques}
 It was shown in \cite{RWZ15} that list decodability of a Gabidulin codes is not beyond the unique decoding bound $\tau=(1-R)/2$. In the classical case of Reed-Solomon codes, the decoding radius can be enlarged by folding Reed-Solomon codes. The question is how to properly fold  Gabidulin codes  to enlarge decoding radius. At the same time, we have to make use of linearized polynomials in order to correct rank errors. Our key idea is to employ two-variable polynomials $f(x,y)$, where $f$ is linearized in variable $x$ and the variable $y$ is used to fold the code. In other words, rows are used to correct rank errors and columns are used to fold the code to enlarge decoding radius.

 The algebraic idea enables us to pin down the messages
into a structured subspace of dimension linear in the number $n$ of columns and this  ``periodic" structure allows us to pre-encode the messages to prune down the list. Two approaches are employed to pin down our list, namely
 subspace design introduced in \cite[STOC2013]{Guru2013}  and hierarchical subspace-evasive (h.s.e. for short) sets introduced in \cite[STOC2012]{Guru2012}.
 The coefficients of polynomials in the list form a ``periodic" subspace. After pre-encoding with  subspace design or h.s.e., the new list becomes a constant.

\subsection{Organization}
The paper is organized as follows. In Section 2, we provide a new construction of ``folded" rank-metric codes and discuss their parameters. Section 3 devotes to list decoding of the rank-metric codes in Section 2, including establishment of interpolation polynomial, solving of certain equations for list and discussion of decoding radius. In the last section, we make use of subspace design and hierarchical subspace-evasive sets to pre-encode the messages and pin down the list. The algorithm from subspace design is deterministic, while the algorithm from hierarchical subspace-evasive sets is Monte Carlo.

\section{Construction of rank-metric codes}
\subsection{Rank-metric codes}
Before introducing our construction, we review some basic facts and results on rank-metric code.

Let $q$ be a prime power and denote by $\MM_{n\times t}(\F_q)$ the set of $n\times t$ matrices over $\F_q$. One can define the rank distance between two matrices $A, B\in \MM_{n\times t}(\F_q)$ to be the rank of $A-B$, i.e., $d(A,B)=\rank(A-B)$. Indeed this defines a distance \cite{Ga85}. A rank-metric code $\mC$ is a subset of $\MM_{n\times t}(\F_q)$
 with rate and distance given by
\[R(\mathcal{C}) = \frac{\log_q\abs{\mathcal{C}}}{nt}\quad \text{ and }\quad d(\mathcal{C}) = \min_{A\neq B\in\mathcal{C}}\{d(A,B)\}.\]
Without loss of generality, from now on we may assume that $n\le t$ (otherwise, we can consider transpose of matrices). As in the classical case, one has the following Singleton bound (see \cite{Ga85})
\begin{equation}\label{eq:a1} d(\mC)\le n-R(\mC)n+1.
\end{equation}
A code archiving the above Singleton bound is called Maximal Rank Distance (or MRD for short) code.  The most famous MRD codes are Gabidulin codes which are defined by using polynomial evaluations. Recently, some MRD codes other than Gabidulin codes have been constructed \cite{S15}.

To better understand our codes, we briefly review the construction of Gabidulin codes \cite{Ga85}. A polynomial of the form $f(x)=\sum_{i=0}^{\ell}a_ix^{q^i}$ is called  $q$-linearized, where coefficients $a_i$ belong to the algebraic closure of $\F_q$. The $q$-degree of $f(x)$, denoted by $\deg_q(f)$, is defined to be $\ell$ if $a_{\ell}\not=0$.

 Let $0<k\leq n\leq t$ be integers, and choose $\F_q$-linearly independent elements $\alpha_1,\dotsc, \alpha_n\in\F_{q^t}$. For every $q$-linearized polynomial $f\in \F_{q^t}[X]$ of $q$-degree at most $k-1$, we can encode $f$ by the column vector $A_f=\bigl( f(\alpha_1),\dotsc, f(\alpha_n)\bigr)^T$ over $\F_{q^t}$. By fixing a basis of $\F_{q^t}$ over $\F_q$, we can also think of $A_f$ as an $n\times t$ matrix over $\F_q$. This yields the Gabidulin code
\[\mathcal{C}_G(q,n,t,k) := \{A_f\in \mathbb{M}_{n\times t}(\F_q):\; f\in\F_{q^t}[x]\ \mbox{is $q$-linearized and} \ \deg_q(f)\leq k-1\}.\]
The  Gabidulin codes are similar to the classical Reed-Solomon codes. However, if applying Sudan's list decoding idea to decoding of the Gabidulin codes, we get only unique decoding (see \cite{KK08}).

In order to enlarge list decoding radius of the Gabidulin codes, Mahdavifar and Vardy \cite{MV12} considered folded Gabidulin codes. As a result, the rate tends to $0$. In the next subsection, we consider evaluations of two-variable polynomials to obtain rank-metric codes with good list decodabiity.

\subsection{Construction}
Let us fix some notations at the beginning.
Let $n,m$ be positive integers with $m\le n$ ($m$ and $n$ are propositional and both  tend to $\infty$). Let $r$ be a prime power and choose a positive integer $k$ with $ k\le r-1$ (both $r$ and $k$ are constant and independent of $n,m$). Put $q=r^{\ell}$ for some $\ell$ with $\gcd(r-1, n\ell)=1$ ($\ell$ is a constant and hence $q$ is a constant as well). Fix a primitive element $\Gg$ of $\F_r^*$.

We have the following  facts:
\begin{itemize}
\item $x^{r-1}-\Gg$ is irreducible over $\F_r$, and hence it is irreducible over $\F_{q^n}$ as well since $\gcd(r-1, n\ell)=1$.
\item $x^r\equiv \Gg x\mod{x^{r-1}-\Gg}$.
\end{itemize}

Consider the two-variable polynomial space over $\F_{q^n}$
\[\mP_q(n,k,m)[x,y]:=\left\{\sum_{i=0}^{m-1}f_i(x)y^{q^i}:\; f_i(x)\in\F_{q^n}[x]\ \mbox{and $\deg(f_i(x))\le k-1$ for all $0\le i\le m-1$}\right\}.\]
Let  $\{\Ga_1,\Ga_2,\dots,\Ga_n\}$ be an $\F_q$-basis of $\F_{q^n}$.  For each polynomial $f=\sum_{i=0}^{m-1}f_i(x)y^{q^i}\in \mP_q(n, k,m)[x,y]$, we define a matrix
\[M_f:=\left(\begin{array}{ccccc}
f(1,\Ga_1)& f(\Gg,\Ga_1) & f(\Gg^2,\Ga_1)&\cdots&f(\Gg^{r-2},\Ga_1)\\
f(1,\Ga_2)& f(\Gg,\Ga_2) & f(\Gg^2,\Ga_2)&\cdots&f(\Gg^{r-2},\Ga_2)\\
\cdots& \cdots & \cdots&\cdots&\cdots\\
f(1,\Ga_n)& f(\Gg,\Ga_n) & f(\Gg^2,\Ga_n)&\cdots&f(\Gg^{r-2},\Ga_n)\\
\end{array}
\right)\]
Each entry in the above matrix is viewed as a row vector of $\F_q^n$. Thus, $M_f$ is an $n\times ((r-1)n)$ matrix over $\F_q$. Set $t=(r-1)n$. Let $\mC_q(n,k,m,r)$ be the collection of $M_f$ for all $f\in \mP_q(n, k,m)[x,y]$.
\begin{lemma} \label{lem:2.1}
The distance and rate of $\mC_q(n,k,m,r)$ satisfy
\[d(\mC_q(n,k,m,r))\ge n-m+1\quad \mbox{and}\quad \quad R(\mC_q(n,k,m,r)):=\frac{\log_qq^{nkm}}{(r-1)n^2}=\frac{k}{r-1}\times\frac{m}{n},\]
respectively.
\end{lemma}
\begin{proof} The size of $\mP_q(n, k,m)[x,y]$ is $q^{nkm}$. Furthermore, it is easy to see that $\mC_q(n,k,m,r)$ is an $\F_q$-linear space. Hence it is sufficient to show that the rank of $M_f$ is at least $n-m+1$ for every nonzero polynomial $f(x,y)\in \mP_q(n, k,m)[x,y]$.

Let $f=\sum_{i=0}^{m-1}f_i(x)y^{q^i}$ in $\mP_q(n, k,m)[x,y]$ be a nonzero polynomial.  Suppose  that $M_f$ has rank less than $n-m+1$. Then the solution space $U\subseteq\F_q^n$ of $\bz M_f=\bo$ has dimension at least $m$.  Let $V$ be the $\F_q$-subspace of $\F_{q^n}$ given by $V=\{\sum_{i=1}^nu_i\Ga_i:\; (u_1,u_2,\dots,u_n)\in U\}$. Then $\dim_{\F_q}(V)=\dim_{\F_q}(U)\ge m$.

 For each $0\le j\le r-2$, Let $g_{j}(y)=f(\Gg^{j},y)$. Then, every $\Ga$ in $V$ is a root of the polynomial $g_{j}(y)$.  Since  $\deg(g_{j}(y))\leq m-1$, the polynomial $f(\Gg^j,y)=g_{j}(y)$ is identical to $0$. This means that the coefficients $f_{i}(\Gg^{j})$ of $g_{j}(y)$ are zero for any $0\le i\le m-1$.  As the degree of $f_{i}(x)$ is at most $k-1$, we conclude that $f_i(x)$ are the zero polynomials for all $0\le i\le m-1$. This is a contradiction and the proof is completed.
\end{proof}

\begin{rmk} The code $\mC_q(n,k,m,r)$ is an MRD code if and only if $k=r-1$.
\end{rmk}
\section{List decoding}
Suppose that a codeword $M_f $ is transmitted and $Y=(y_{i,j})_{1\le i\le n; 0\le j\le r-2}$ is received with at most $e$ errors, i.e., $\rank(M_f-Y)\le e$. Our goal in this section is to recover $M_f$, or equivalently the polynomial $f(x,y)\in \mP_q(n,k,m)[x,y]$. First we prove a lemma on rank of matrices.

\begin{lemma}
\label{lem:3.1}
Let $X,Z\in \MM_{n\times t}(\F_q)$ with $\mathrm{rank}(X-Z)\le e$. Then $\dim_{\F_q}(\langle X\rangle\cap\langle Z\rangle)\ge \dim_{\F_q}(\langle X\rangle)-e$, where $\langle X\rangle$ stands for the row space of $X$ over $\F_q$.
\end{lemma}
\begin{proof} It is easy to see that the two $\F_q$-spaces   $\langle X\rangle+\langle Z\rangle$ and $\langle X-Z\rangle+\langle Z\rangle$ are equal. Thus,
\[\dim_{\F_q}(\langle X\rangle)+\dim_{\F_q}(\langle Z\rangle)-\dim_{\F_q}(\langle X\rangle\cap \langle Z\rangle)=\dim_{\F_q}(\langle X-Z\rangle)+\dim_{\F_q}(\langle Z\rangle)-\dim_{\F_q}(\langle X-Z\rangle\cap \langle Z\rangle).\]
This gives
\[\dim_{\F_q}(\langle X\rangle\cap \langle Z\rangle)=\dim_{\F_q}(\langle X\rangle)-\dim_{\F_q}(\langle X-Z\rangle)+\dim_{\F_q}(\langle X-Z\rangle\cap \langle Z\rangle)\ge \dim_{\F_q}(\langle X\rangle)-e.\]
The proof is completed.
\end{proof}

\subsection{Interpolation polynomials}


We fix a parameter $s$ with $1\le s\le r-1$.

\begin{defn}[Space of interpolation  polynomials]
Let $\mathcal{L}$ be the space of  polynomials $Q \in \F_{q^n}[x,y,z_1,z_2,$ $\dots,z_s]$ of the form
$Q(x, y,z_1,z_2,\dots,z_s) =
A_0(x,y) + A_1(x,z_1)  + A_2(x,z_2)  + \cdots + A_s(x,z_s)$,
with $A_0(x,y)\in \mP_q(n, r-1, n-e)[x,y]$ and
each $A_i(x,z_i) \in \mP_q(n, r-k, n-e-m+1)[x,z_i]$ for $i=1,2,\dots,s$.
\end{defn}
\begin{lemma}
\label{lem:3.2} If $e<\frac{s(r-k)(n-m+1)}{r-1+s(r-k))}$, then
there exists a nonzero polynomial $Q \in \mathcal{L}$ such that $Q(\Gg^{j},\Ga_i,y_{i,j},$ $y_{i,j+1},\ldots,y_{i,j+s-1}) = 0$ for $i=1,2,\dots,n$ and $j=0,1,2,\dots,r-2$. Note that if $j+s-1$ is bigger than $r-2$, we replace $y_{i,j+s-1}$ by $y_{i,j+s-1\bmod{r-1}}$.
Furthermore, such a polynomial $Q$ can be found using $O(n^4)$ operations over $\F_{q^n}$.
\end{lemma}
\begin{proof}
Note that $\mathcal{L}$ is an $\F_{q^n}$-vector space of dimension $(r-1)(n-e)+s(r-k)(n-e-m+1)$. This dimension is bigger than $n(r-1)$ by our choice of $m$ and $k$.
The conditions to be satisfied in the Lemma give rise to $n(r-1)$ homogeneous linear conditions on $Q$. Since  $n(r-1)<(r-1)(n-e)+s(r-k)(n-e-m+1)$ in our setting, there must exist a nonzero $Q \in \mathcal{L}$ that meets the interpolation conditions $Q(\Gg^{j},\Ga_i,y_{i,j},y_{i,j+1},y_{i,j+2},\cdots,y_{i,j+s-1}) = 0$ for $i=1,2,\dots,n$ and $j=0,1,\dots,r-2$.
Finding such a polynomial $Q$ amounts to solving a homogeneous linear system over $\F_{q^n}$ with $n(r-1)$ constraints and ${\dim}_{\F_{q^n}}(\mathcal{L})=(r-1)(n-e)+s(r-k)(n-e-m+1)$ unknowns, which can be done in $O(n^4)$ time.
\end{proof}

\begin{lemma}\label{lem:3.3} Let $f\in \mP_q(n, k, m)[x,y]$  be a  polynomial. Suppose that the codeword $M_f $ is transmitted and $Y=(y_{i,j})_{n\times (r-1)}$ ($y_{i,j}\in \F_{q^{n}}$) is received with at most $e$ errors. Assume that $e<\frac{s(r-k)(n-m+1)}{r-1+s(r-k)}$ and let $Q(x, y,z_1,z_2,\dots,z_s)$ be the interpolation polynomial given in Lemma \ref{lem:3.2}. Then
\begin{equation}\label{eq:4.1}Q(\Gg^{j}, y, f(\Gg^{j},y), f(\Gg^{j+1},y), f(\Gg^{j+2},y), \cdots, f(\Gg^{j+s-1},y)) \equiv 0\end{equation}
for all $j=0,1,2,\dots, r-2$. The above $\equiv$ means that the polynomial on the left is identical to $0$.
\end{lemma}
\begin{proof}
Note that $e<\frac{s(r-k)(n-m+1)}{r-1+s(r-k))}<n-m+1$. Since $e$ and $n-m$ are both integers, we have $e\leq n-m$.
The polynomial $Q(\Gg^{j}, y, f(\Gg^{j},y), f(\Gg^{j+1},y), f(\Gg^{j+2},y), \cdots, f(\Gg^{j+s-1},y)) $ has degree at most $q^{m-1}$, moreover it is $q$-linearized. Denote by $A$ and $B$ the $n\times rn$ matrices $((\Ga_1,\Ga_2,\dots,\Ga_n)^T, M_f)$ and $((\Ga_1,\Ga_2,\dots,\Ga_n)^T, Y)$ over $\F_q$, respectively.

It is clear that ${\rm rank}(A-B)={\rm rank}(M_f-Y)\le e$ and ${\rm rank}(A)=n$. Thus, by Lemma \ref{lem:3.1} $\dim_{\F_q}(\langle A\rangle\cap \langle B\rangle)\ge n-e\geq m$. This implies that  exists an $\F_q$-subspace $U$ of ${\rm span}\{\Ga_1,\Ga_2,\dots,\Ga_n\}$ of dimension at least $m$  such that, for every $\Ga=\sum_{i=1}^nc_i\Ga_i\in U$ with $c_i\in\F_q$, one has
 \[\sum_{i=1}^nc_iy_{i,j+u-1}=\sum_{i=1}^nc_if(\Gg^{j+u-1},\Ga_i)=f\left(\Gg^{j+u-1},\sum_{i=1}^nc_i\Ga_i\right)=f(\Gg^{j+u-1},\Ga)\] for $u=1,2,\dots,s$.
Hence,
\begin{eqnarray*}0&=&\sum_{i=1}^nc_iQ(\Gg^{j},\Ga_i,y_{i,j},y_{i,j+1},\cdots,y_{i,j+s-1}) \\
&=&\sum_{i=1}^{n}\left(c_{i}A_{0}(\Gg^{j},\Ga_{i})+\sum_{u=1}^{s}c_{i}A_{u}(\Gg^{j},y_{i,j+u-1})\right)\\
&=&A_{0}\left(\Gg^{j},\sum_{i=1}^{n}c_{i}\Ga_{i}\right)+\sum_{u=1}^{s}A_{u}\left(\Gg^{j},\sum_{i=1}^{n}c_{i}y_{i,j+u-1}\right)\\
&=&A_{0}\left(\Gg^{j},\Ga\right)+\sum_{u=1}^{s}A_{u}\left(\Gg^{j},f(\Gg^{j+u-1},\Ga)\right)\\
&=&Q(\Gg^{j}, \Ga, f(\Gg^{j},\Ga), f(\Gg^{j+1},\Ga), f(\Gg^{j+2},\Ga), \cdots, f(\Gg^{j+s-1},\Ga)).\end{eqnarray*}
As the degree of $Q(\Gg^{j}, y, f(\Gg^{j},y), f(\Gg^{j+1},y), f(\Gg^{j+2},y), \cdots, f(\Gg^{j+s-1},y)) $ is at most $q^{m-1}$. The desired result follows.
\end{proof}

\begin{lemma}\label{lem:3.4} Let $f=\sum_{i=0}^{m-1}f_i(x)y^{q^i}\in \mP_q(n, k, m)[x,y]$  be a  polynomial.   Suppose that the codeword $M_f $ is transmitted and $Y$ is received with at most $e$ errors. Assume that $e<\frac{s(r-k)(n-m+1)}{r-1+s(r-k)}$ and let $Q(x, y,z_1,z_2,\dots,z_s)=
A_0(x,y) + A_1(x,z_1)  + A_2(x,z_2)  + \cdots + A_s(x,z_s)$ be the interpolation polynomial given in Lemma \ref{lem:3.2}. Write $A_0(x,y)=\sum_{i=0}^{n-e-1}A_{0,i}(x)y^{q^i}$ and $A_w(x,z)=\sum_{i=0}^{n-e-m}A_{w,i}(x)z^{q^i}$ for $1\le w\le s$. Then we have
\begin{equation}\label{eq:3}A_{0,u}(x)+\sum_{w=1}^s\sum_{i+v=u}A_{w,i}(x)f_v^{(i)}(\Gg^{w-1}x) \equiv0\end{equation}
for all $0\le u\le n-e-1$, where $g^{(j)}(x)$ stands for $\sum_{i=0}^Ng_i^{q^j}x^i$ for a polynomial $g(x)=\sum_{i=0}^Ng_ix^i\in\F_{q^n}[x]$.
\end{lemma}
\begin{proof} By Lemma \ref{lem:3.3}, we have
\begin{eqnarray*}
0&\equiv&Q(\Gg^{j}, y, f(\Gg^{j},y), f(\Gg^{j+1},y), f(\Gg^{j+2},y), \cdots, f(\Gg^{j+s-1},y)) \\
&=& \sum_{u=0}^{n-e-1}A_{0,u}(\Gg^{j})y^{q^u}+\sum_{w=1}^s\sum_{i=0}^{n-e-m}A_{w,i}(\Gg^{j})\left(\sum_{v=0}^{m-1}f_v(\Gg^{w+j-1})y^{q^v}\right)^{q^i}\\
&=& \sum_{u=0}^{n-e-1}A_{0,u}(\Gg^{j})y^{q^u}+\sum_{u=0}^{n-e-1}\left(\sum_{w=1}^s\sum_{i+v=u}A_{w,i}(\Gg^{j})f_v^{(i)}(\Gg^{w+j-1})\right)y^{q^u}\\
\end{eqnarray*}
This gives
\begin{equation*}A_{0,u}(\Gg^{j})+\sum_{w=1}^s\sum_{i+v=u}A_{w,i}(\Gg^{j})f_v^{(i)}(\Gg^{w+j-1})=0
\end{equation*}
for all $0\le u\le n-e-1$ and $0\le j\le r-2$. This implies that the polynomial \[A_{0,u}(x)+\sum_{w=1}^s\sum_{i+v=u}A_{w,i}(x)f_v^{(i)}(\Gg^{w-1}x)\] has at least $r-1$ roots. On the other hand, this polynomial has degree at most $k-1\le r-2$. The desired result follows.
\end{proof}

\subsection{Analysis of list and list size}
Before discussing the list, let us introduce periodic subspaces that were defined in \cite{Guru2012}. For a vector $\mv{a}=(a_1,a_2,\dots,a_N) \in \F_r^N$ and positive
integers $t_1 \le t_2 \le m$, we  denote by
$\proj_{[t_1,t_2]}(\mv{a}) \in \F_q^{t_2-t_1+1}$ its projection onto
coordinates $t_1$ through $t_2$, i.e.,
$\proj_{[t_1,t_2]}(\mv{a})=(a_{t_1},a_{t_1+1},\dots,a_{t_2})$. When
$t_1=1$, we use $\proj_t(\mv{a})$ to denote
$\proj_{[1,t]}(\mv{a})$. These notions are extended to subsets of
strings in the obvious way: $\proj_{[t_1,t_2]}(S) = \{
\proj_{[t_1,t_2]}(\mv{x}) :\; \mv{x} \in S\}$.

\begin{defn}[Periodic subspaces]
\label{def:periodic-subspaces}
For positive integers $u,b,\period$ and $\dims := b\period$, an affine subspace $H \subset \F_r^\dims$ is said to be $(u,\period,b)_r$-periodic if there exists a subspace $W \subseteq \F_r^\period$ of dimension at most $u$ such that
 for every $j=1,2,\dots,b$, and every ``prefix" $\mv{a} \in \F_q^{(j-1)\period}$,  the projected
 affine subspace of $\F_r^\period$ defined as
\[ \{ \proj_{[(j-1)\period+1,j\period]}(\mv{x}):\;\mv{x} \in H \mbox{ and } \proj_{(j-1)\period}(\mv{x}) =\mv{a} \} \  \]
is contained in an affine subspace of $\F_r^\period$ given by $W + \mv{v}_{\mv{a}}$ for some vector $\mv{v}_{\mv{a}} \in \F^\period$ dependent on $\mv{a}$.
\end{defn}

Now we return to finding list of polynomial candidates.
\begin{lemma}\label{lem:3.5} Let $f=\sum_{i=0}^{m-1}f_i(x)y^{q^i}\in \mP_q(n, k, m)[x,y]$  be a  polynomial.   Suppose that the codeword $M_f $ is transmitted and $Y$ is received with at most $e$ errors. Assume that $e<\frac{s(r-k)(n-m+1)}{r-1+s(r-k)}$. Then solutions of \eqref{eq:3} form  an $(s-1,\ell n(r-1),m)_r$-periodic subspace of size at most $r^{m(s-1)}$.
\end{lemma}
\begin{proof}
Note that for $u\in[0,n-e-1]$, the solutions of
\eqref{eq:3} give the list of the candidates.

Let us start with $u=0$. Then \eqref{eq:3} gives the equation
\begin{equation}\label{eq:4}
A_{0,0}(x)+\sum_{w=1}^{s}A_{w,0}(x)f_{0}^{(0)}(\gamma^{w-1}x)=0
\end{equation}
Note that  $f_{0}^{(0)}(x)=f_0(x)$. In the residue ring $\F_{q^n}[x]/(x^{r-1}-\Gg)$, the equation \eqref{eq:4} becomes
\begin{equation}\label{eq:5}
A_{0,0}(x)+\sum_{w=1}^{s}A_{w,0}(x)(f_{0}(x))^{r^{w-1}}\equiv0\quad \bmod{x^{r-1}-\gamma}.
\end{equation}
Since $x^{r-1}-\gamma$ is an irreducible polynomial over $\F_{q^n}$, the residue ring $\F_{q^n}[x]/(x^{r-1}-\Gg)\simeq\F_{q^{n(r-1)}}$ is a field. Because the degree of $f_{0}(x)$ is at most $r-2$, all solutions of $f_{0}(x)$ in the equation \eqref{eq:5} form an affine space $W+\bv_1$ for some $\bv_1\in\F_{q^n}[x]/(x^{r-1}-\Gg)\simeq\F_r^{\ell n(r-1)}$, where $W$ is the solution space  of the $\F_r$-linearized polynomial
\begin{equation}\label{eq:6}
\sum_{w=1}^{s}A_{w,0}(x)z^{r^{w-1}}\equiv 0\quad \bmod{x^{r-1}-\gamma}\end{equation} and therefore it has dimension at most $s-1$ over $\F_r$.

Note that once $f_{0}(x)$ is recovered, all $f_0^{(j)}$ are recovered as well for $j\ge 0$.

By induction, assume that all $f_{i}(x)$ have been recovered for $0\le i\le a-1$. Next, we want to recover $f_{a}(x)$ from the following equation
\[
 A_{0,a}(x)+\sum_{w=1}^{s}\sum_{i+v=a}A_{w,i}(x)(f_{v}^{(i)}(x))^{r^{w-1}}\equiv0\quad \bmod{x^{r-1}-\gamma}.
\]
Rewrite the above equation into the following
\begin{equation}\label{eq:7}
 A_{0,a}(x)+\sum_{w=1}^s\sum_{v=1}^{a-1}A_{w,a-v}(x)(f_{v}^{(a-v)}(x))^{r^{w-1}}   +\sum_{w=1}^{s}A_{w,0}(x)(f_{a}^{(0)}(x))^{r^{w-1}}\equiv0\quad \bmod{x^{r-1}-\gamma}
\end{equation}
By the similar arguments, one can show that all solutions of $f_{a}^{(0)}(x)=f_a(x)$ in the equation \eqref{eq:7} form an affine space $W+\bv_a$ for some $\bv_a\in\F_{q^n}[x]/(x^{r-1}-\Gg)\simeq\F_r^{\ell n(r-1)}$.
Apparently, all possible  $(f_0(x),f_1(x),\dots,f_{m-1}(x))$ in the list form an $(s-1,\ell n(r-1),m)_r$-periodic subspace.

To compute the list size, we note that each $f_i(x)$ has at most $r^{s-1}$ solutions. Thus, the list size is bounded by $r^{m(s-1)}$.
\end{proof}
As $m$ is promotional to $n$, the list size $r^{m(s-1)}$ in Lemma \ref{lem:3.5} becomes exponential. We will prune down the list size by pre-encoding through the special structure of periodic subspace.
\begin{rmk}\label{rmk:2} Each $f_i(x)$ is a solution of \eqref{eq:7}. As $\deg(f_i(x))\le k-1$, there exist an $g(x)\in\F_{q^n}[x]$ with $\deg(g(x))\le k-1$ such that $f(x)\in g(x)+W'$, where $W'=W\cap\{h(x)\in\F_{q^n}[x]:\;\deg(h)\le k-1\}$ and  $W$ is the solution space of \eqref{eq:6}. This implies that our message $f(x)$ actually belongs to an $(s-1,\ell nk,m)_r$-periodic subspace of size at most $r^{m(s-1)}$.
\end{rmk}

\subsection{Decoding radius}
Finally, let us compute the decoding radius from the list decoding in this section.

Put $e=\left\lfloor\frac{s(r-k)(n-m+1)}{r-1+s(r-k))}\right\rfloor-1$ and $\tau=e/ n$, then we have
\begin{equation}\label{eq:8}\tau\thickapprox \frac{s(r-k)}{r-1+s(r-k)}\left(1-\frac mn\right)=\frac{s(r-k)}{r-1+s(r-k)}\left(1- R\times\frac{r-1}k\right).\end{equation}

If we take $s=r-1$ and $k=r-c$ for some $1\le c\le r-1$, then we get
 \begin{equation}\label{eq:11}\tau=\frac{c}{c+1}\left(1-\frac{r-1}{r-c}\times R\right).\end{equation}

\section{Pruning list size}
In this section, we prune list via subspace design and h.s.e. The subspace design provides a deterministic algorithm with a constant list size, while h.s.e  provides a randomized algorithm with a smaller constant list size.

\subsection{A deterministic algorithm}
The subspace design was first introduced in \cite{Guru2013} to pin down list.

\begin{defn} A collection $S$ of $\F_r$-subspaces $H_1,\dotsc, H_M\subseteq\F_r^{\GL}$ is called a $(v,A,\GL)_r$-{subspace design} if for every $\F_{r}$-linear space $W\subset \F_r^{\GL}$ of dimension $v$,
\[\sum_{i=1}^M \dim_{\F_r} (H_i\cap W)\leq A.\]
\end{defn}

In order to pin down the list to a constant size, one has to consider intersection with subspace evasive set introduced in \cite{GW}.

\begin{defn} A subset $S$ of $\F_r^{\GL}$ is called a $(v,A,\GL)_r$-subspace evasive if for any subspace $W$ of $\F_r^{\GL}$ of dimension $v$, the intersection $S\cap W$ has size at most $A$.
\end{defn}

The following result tells that one can obtain a  small list from intersection of a periodic subspace with a suitable subspace design .

\begin{lemma}(\cite{Guru2013,Guru2015})\label{lem:4.1}
Let $H$ be a $(v, \Lambda, b)_r$-periodic subspace, and let $\{H_1,H_2,\dots,H_b\}$  be a $(v, A,\GL)_r$-subspace design. Then $H\cap (H_1\times\dotsb\times H_b)$ is an affine subspace over $\F_r$ of dimension at most $A$.
\end{lemma}

Assume that $\Lambda$ has a divisor $\lambda\thickapprox 2\log_r \Lambda$ for some $c>1$ and thus we have $r^{\lambda}>\Lambda$. Let $q_{1}=r^{\lambda}$ and $\Lambda'=\Lambda/ \lambda$.

\begin{lemma}[\cite{Dvir2012}]
\label{lem:4.2} Let $\Ge>0$ be a small real. Let $v$ be a positive integer and set $h\thickapprox v/\Ge$ to be a positive integer. Assume that $q_1\ge h$  and let $\gamma_1,\dotsc, \gamma_h$ be distinct nonzero elements of $\F_{q_1}$.  Let $d_1>d_2>\dotsb > d_h\geq 1$ be integers. Define $f_1,\dotsc, f_v\in \F_{q_1}[x_1,\dotsc, x_h]$ as follows:
\begin{equation}
\label{eqn:dl}
f_i(x_1,\dotsc, x_h) = \sum_{j=1}^h \gamma_j^i x_j^{d_j} \ .
\end{equation}
Then:
\begin{itemize}
\item The variety $\mathbf{V}=\{\bx\in \overline{\F}_{q_1}^h\mid f_1(\bx)=\dotsb= f_v(\bx)=0\}$ satisfies $\abs{\mathbf{V}\cap H}\leq (d_1)^v$ for all $v$-dimensional affine subspaces $H\subset \overline{\F}_{q_1}^h$.
\item
If at least $v$ of the degrees $d_i$ are relatively prime to $q_1-1$, then $\abs{\mathbf{V}\cap \F_{q_1}^h}= q_1^{h-v}$.
\item The product set $(\mathbf{V}\cap \F_{q_1}^h)^{\GL'/h}\subseteq \F_{q_1}^{\GL'}$ is $(a, (d_1)^a,\GL')_{q_1}$-subspace evasive for all $a\leq v$.
\end{itemize}
\end{lemma}

The below statement follows immediately from Lemma~\ref{lem:4.2} and the fact that when the $d_j$'s are powers of $r$, the polynomials $f_i$ defined in \eqref{eqn:dl} are $\F_r$-linearized polynomials.

\begin{cor}\label{col:1} Let $\Ge>0$ be a small real. Let $v$ be a positive integer and set $h\thickapprox v/\Ge$ to be a positive integer. Assume that $q_1\ge h$.
By setting $d_{1}=r^{h-1},d_{2}=r^{h-2},\ldots,d_{h}=1$ in Lemma \ref{lem:4.2}, one obtains an explicit $(a,r^{a(h-1)},\GL')_{q_1}$-subspace evasive set $S$ of  size $q_{1}^{(1-\epsilon)\Lambda'}$  for all $1\leq a\leq v$. Furthermore, $S$ is an $\F_r$-linear space of dimension $(1-\Ge)\Gl\GL'=(1-\Ge)\GL$ and a basis of $S$ can be computed in time ${\rm poly}(\GL,\log r)$.
\end{cor}
Guruswami and Kopparty~\cite{Guru2013B} gives an explicit subspace design based on Wronskian determinant. Their construction implies the following fact.

\begin{lemma}\label{lem:4.4}\label{thm:subspace}
For $\epsilon\in (0,1)$, positive integer $v$  with $v<\epsilon \Lambda'/4$, there is an explicit collection of $M=q_{1}^{\Omega(\epsilon \Lambda'/v)}$ subspaces in $\F_{q_{1}}^{\Lambda'}$, each of codimension at most $\epsilon \Lambda'$ and form a $(v,2v/\epsilon,\GL')_{q_1}$-subspace design. Moreover, bases for $N\le M$ elements of this collection can be computed in time ${\rm poly}(N,\GL, r)$.
\end{lemma}
It is required in Lemma \ref{lem:4.4} that $q_1>\GL'$ (see \cite{Guru2013B}).  This condition is satisfied by our choice of parameters since $q_{1}=r^{\lambda}>\Lambda$.

Combined Lemma \ref{lem:4.4} with Corollary~\ref{col:1}, one can prove the following result.
\begin{prop}\label{prop:4.5}
For a positive integer $v\leq \epsilon\Lambda'/4$, there exists an explicit  $(v,2v(h-1)/\Ge,\GL)_r$-subspace design $\{H_1,H_2,\dots,H_N\}$ with $N=q_{1}^{\Omega(\epsilon \Lambda'/v)}$ and $H_i\subseteq \F_{q_{1}}^{\Lambda'}=\F_{r}^{\Lambda}$ of codimension at most $2\epsilon \Lambda$ .
\end{prop}
\begin{proof} The proof of this proposition can be found in \cite[Theorem 3.6]{Guru2015} except for adjustment of parameters.
To convince the reader of that our parameters work properly, we give a complete proof  here. From Lemma~\ref{thm:subspace}, we can construct $M=q_{1}^{\Omega(\epsilon \Lambda'/s)}$ subspaces $V_1,V_2,\ldots,V_M$ with codimension at most $\epsilon \Lambda'$ over $\F_{q_{1}}$. By Corollary~\ref{col:1}, we know that there exists an explicit $\F_{r}$-linear space $S$ of size $q_{1}^{(1-\epsilon)\Lambda'}$ in $\F_{q_{1}}^{\Lambda'}$ which is $(a,h^{a(\ell-1)},\GL')_{q_1}$-subspace evasive for $a\leq v$.
Put $H_{i}=V_{i}\cap S$. Since both $V_{i}$ and $S$ has codimension at most $\epsilon\Lambda'$ in $F_{q_1}^{\GL'}$, the intersection $H_{i}$ has codimension at most $2\epsilon\Lambda'$ in $\F_{q_1}^{\GL'}$, i.e., $H_{i}$ has codimension  at most $2\epsilon\Lambda$ in $\F_{r}^{\GL}$. Let $W$ be a $v$-dimensional $\F_r$-linear subspace in $\F_{q_1}^{\GL'}$. Then one can find a $v$-dimensional $\F_{q_1}$-linear subspace $W_1$ in $\F_{q_1}^{\GL'}$ such that $W\subseteq W_1$.

The subspace design of $\{V_{i}\}_{i=1}^M$ implies that
\begin{equation}\label{eq:13}  \sum_{i=1}^{M}\dim_{\F_{q_{1}}}(V_{i}\cap W_1)\leq 2 v/\epsilon
\end{equation}
Denote by $v_{i}$ the dimension  $\dim_{\F_{q_{1}}}(V_{i}\cap W_1)$. As $\dim_{\F_{q_{1}}}(W_1)\leq v$, we have that $v_{i}\leq v$. Since $S$ is a $(v_i,r^{v_i(h-1)},\GL')_{q_1}$-subspace evasive set, we have $|S\cap(V_{i}\cap W_1)|\le r^{v_i(h-1)}$. Hence, $\dim_{\F_{r}}(H_{i}\cap W_1)\le v_i(h-1)=(h-1) \dim_{\F_{q_{1}}}(V_{i}\cap W_1)$.
Summing all dimensions up gives
\begin{equation*}
\sum_{i=1}^{M}\dim_{\F_{r}}(H_{i}\cap W)\le\sum_{i=1}^{M}\dim_{\F_{r}}(H_{i}\cap W_1)\le(h-1)  \sum_{i=1}^{M}\dim_{\F_{q_{1}}}(V_{i}\cap W_1)\leq 2 v(h-1)/\epsilon.
\end{equation*}
The proof is completed.
\end{proof}
\begin{theorem}{\bf [Part (i) of Main Theorem]}\label{thm:4.6} Let $r$ be a prime power and let $c$ be an integer between $1$ and $r-1$. Let $\widetilde{\Ge}>0$ be a small real. Let $q=r^{\ell}$ with $\gcd(r-1,\ell n)=1$. Then there exists an explicit rank-metric code in $\MM_{n\times(r-1)n}(\F_q)$ with rate $\widetilde{R}$ that is $(\widetilde{\tau}, O(\exp(1/\widetilde{\Ge}^2)))$-list decodable with $\widetilde{\tau}=\frac c{c+1}\left(1-\frac{r-1}{r-c}\times \widetilde{R}-\widetilde{\Ge}\right)$. Furthermore, encoding and list-decoding algorithms are in polynomial time ${\rm poly}(n,\exp(1/\widetilde{\Ge}))$.
\end{theorem}
\begin{proof} In Proposition \ref{prop:4.5}, we set $v=s-1$, $\GL=n\ell(r-1)$ and $h\thickapprox (s-1)/\Ge$. Each $H_i$ can be viewed as an $\F_r$-subspace of  the polynomial space $\{g(x)\in\F_{q^n}[x]:\;\deg(g(x))\le r-1\}$.

We consider the polynomial set
\[\widetilde{\mP}_q(n,k,m)[x,y]:=\left\{\sum_{i=0}^{m-1}f_i(x)y^{q^i}:\; f_i(x)\in H_i\ \mbox{and $\deg(f_i(x))\le k-1$ for all $0\le i\le m-1$}\right\}.\]
 and the code $\widetilde{\mC}_q(n,k,m,r)=\{M_f:\; f\in \widetilde{\mP}_q(n, k,m)[x,y]\}$. It is clear that $\widetilde{\mC}_q(n,k,m,r)$ is $\F_r$ -linear and it is a subcode of our original code ${\mC}_q(n,k,m,r)$. It is easy to see that
  \begin{equation}\label{eq:14}
\dim_{\F_r}(\widetilde{\mP}_q(n,k,m)[x,y])\ge \sum_{i=0}^{m-1}\dim_{\F_r}(H_i\cap\{f_i(x)\in\F_{q^n}[x]:\; \deg(f_i)\le k-1\})\ge m(n\ell k-2\Ge\GL).
 \end{equation}

  By \eqref{eq:14}, the rate $\widetilde{R}$ of $\widetilde{\mC}_q(n,k,m,r)$ is lower bounded by
 \begin{equation}\label{eq:15}
  \widetilde{R}=\frac{\log_q|\widetilde{\mP}_q(n,k,m)[x,y]|}{(r-1)n^2}\ge \frac{k}{r-1}\times\frac{m}{n}-2\Ge\times \frac mn\ge R-2\Ge.
 \end{equation}
Suppose a codeword $M_f$ with $f\in  \widetilde{\mP}_q(n,k,m)[x,y]$ was transmitted and $Y$ is received with at most $e$ errors, where $e< \frac{s(r-k)(n-m)}{r-1+s(r-k)}$. Then all list belong to the solution space $H$
of \eqref{eq:3} which is an $(s-1,\ell n(r-1),m)_r$-periodic subspace. By Lemma \ref{lem:4.1} and Proposition \ref{prop:4.5}, the list size for the code $\widetilde{\mC}_q(n,k,m,r)$ is $r^{O(s^2/\Ge^2)}=\exp(O(s^2/\Ge^2))=\exp(O(1/\Ge^2))$.

The decoding radius of $\widetilde{\mC}_q(n,k,m,r)$ is equal to those of ${\mC}_q(n,k,m,r)$. By \eqref{eq:11}, we have
\[\widetilde{\tau}=\tau\approx\frac c{c+1}\left(1-\frac{r-1}{r-c}\times R\right)\ge \frac c{c+1}\left(1-\frac{r-1}{r-c}\times \widetilde{R}-\frac{r-1}{r-c}\times 2\Ge\right)\]
for $1\le c\le r-2$. Setting $\widetilde{\Ge}=\frac{r-1}{r-c}\times 2\Ge$ gives the desired result.
\end{proof}

\begin{rmk}\label{rmk:5} In the code $\widetilde{\mC}_q(n,k,m,r)$, if we set $s\approx 4/\Ge^2$,  $r\approx 4/\Ge^2$ and $k/(r-1)=\Ge/2$, then one gets the list decoding radius $\widetilde{\tau}\approx 1-\widetilde{R}-\widetilde{\Ge}$. In this case, the list size is becomes  $(1/\widetilde{\Ge})^{O(\exp(1/\widetilde{\Ge}^4))}$. This proves Corollary \ref{cor:1.3}(i).
\end{rmk}

\subsection{A Monte Carlo algorithm}
We first define subspace evasive for a particular famyly of affine spaces.
\begin{defn}\cite{Guru2013}
Let $\mF$ be a family of affine subspace of $\F_r^{\Gk}$ and each of subspace in $\mF$ has dimension at most $v$. A subset $S\subset  \F_r^{\Gk}$ is called $(\mF,v,\Gk,L)_r$-evasive if $|S\cap W|\le L$ for every $W\in\mF$.
\end{defn}

Now we are able to state our randomized result. The {\bf HSE} map below is actually defined from hierarchical subspace-evasive sets (see \cite{Guru2012,Guru2013}).

\begin{prop}\label{prop:4.7}
Suppose $b,\Lambda,v,\Ga$ are positive integers  and $\zeta$ satisfies the conditions $b\geq (\Ga+1)/\zeta$ and $\Lambda >\frac{2s(\Ga+2)}{\zeta}$. Let $\mF$ be a family of $(v,\Lambda,b)$-periodic subspaces of $\F_{r}^{\Gk}$ with $|\mF|\leq r^{\Ga\Gk}$, where $\Gk=b\GL$. Then there exists a randomized construction of an injective map $\mathbf{HSE}$: $\F_{r}^{(1-2\zeta)\Gk}\rightarrow \F_{r}^{\Gk}$ in time $poly(m\Lambda,1/\zeta,\log r,v)$   such that with  probability at least $1-2^{\Omega(b\GL)}$, the image of $\mathbf{HSE}$ is an $(\mF,bv, \Gk, \frac{\Ga+1}{\zeta})$-subspace evasive set.  Further, given a $(v,\Lambda,b)$-periodic subspace $H\in \mF$, one can compute the set $\{\mathbf{x}\in \F_{r}^{(1-2\zeta)\Gk}:\; \mathbf{HSE}(\mathbf{x})\in H\}$ of size at most $\frac{\Ga+1}{\zeta}$ in deterministic $poly(m\Lambda,r^{v},1/\zeta)$ time.
\end{prop}

\begin{theorem}{\bf [Part (ii) of Main Theorem ]}\label{thm:4.8} Let $r$ be a prime power and let $c$ be an integer between $1$ and $r-1$. Let $\widehat{\Ge}>0$ be a small real. Let $q=r^{\ell}$ with $\gcd(r-1,\ell n)=1$. Then with high probability one can randomly sample a rank-metric code in $\MM_{n\times(r-1)n}(\F_q)$ with rate $\widehat{{R}}$ that is $(\widehat{{\tau}}, O(1/\widehat{\Ge}))$-list decodable with $\widehat{{\tau}}=\frac c{c+1}\left(1-\frac{r-1}{r-c}\times \widehat{R}-\widehat{\Ge}\right)$. Furthermore, encoding and list-decoding algorithms are in polynomial time ${\rm poly}(n,\exp(1/\widehat{\Ge}))$.
\end{theorem}
\begin{proof} In Proposition \ref{prop:4.7}, set $v=s-1$, $b=m$ and $\GL=n\ell k$. Let $\mF$ be the set of all $(s-1,n\ell k,m)_r$-periodic subspaces in $\F_r^{mn\ell k}$. A periodic subspace $H\subseteq \F_r^{mn\ell k}$ consists of a fixed subspace $W\subseteq \F_r^{\GL}$ of dimension at most $s-1$ and affine space $\proj_{[(j-1)\Lambda+1,j\Lambda]}(H)=W+\bv_j$ with $\bv_j\in\F_{q^n}^{k}$ for $j=1,2,\dots,m$. Thus, there are at most $N_s\times r^{m\GL}$ periodic subspaces in $\mF$, where $N_s$ denotes the number of subspaces in $\F_r^{\GL}$ of dimension less than or equal to $s-1$. As $m$ tends to $\infty$ and $s$ is a constant, one clearly has
\[N_s=\sum_{i=0}^{s-1}{\GL\brack i}_r\le s{\GL\brack s-1}_r\le (s-1)r^{(s-1)\GL}\le r^{m\GL},\]
where ${\GL\brack i}_r$ denotes the Gaussian binomial coefficients that is equal to the number of subspaces of $\F_r^{\GL}$ of dimension $i$. Thus, in total we have $|\mF|\le r^{2m\GL}$.

In Proposition \ref{prop:4.7}, we set $\Ga=2$. Let $\mathbf{HSE}$ be the injective map given in Proposition \ref{prop:4.7}:  $\F_{r}^{(1-2\zeta)m\GL}\rightarrow \F_{r}^{m\GL}$. As $\F_{r}^{m\GL}\simeq {\mP}_q(n,k,m)[x,y]$, we can identify these two spaces under a fixed basis and hence $\mathbf{HSE}(\bx)$ can be viewed as  a polynomial in ${\mP}_q(n,k,m)[x,y]$. Now our encoding becomes
\[\F_{r}^{(1-2\zeta)m\GL}\rightarrow \F_{r}^{m\GL}\simeq {\mP}_q(n,k,m)[x,y]\rightarrow\MM_{n\times(r-1)n}(\F_q);\quad \bx\mapsto \mathbf{HSE}(\bx)\mapsto M_{\mathbf{HSE}(\bx)}.\]
Denote by $\widehat{{\mC}}_q(n,k,m,r)$ the image of the above map. Thus the rate of the code $\widehat{{\mC}}_q(n,k,m,r)$ is
\begin{equation}\label{eq:16}
\widehat{R}=\frac{\log_qr^{(1-2\zeta)m\GL}}{n^2(r-1)}=(1-2\Gz)\times \frac{k}{r-1}\times \frac mn=(1-2\Gz)R\ge R-2\Gz,
\end{equation}
where $R$ is the rate of ${{\mC}}_q(n,k,m,r)$.

Suppose a codeword $M_{\mathbf{HSE}(\bx)}$  was transmitted and $Y$ is received with at most $e$ errors, where $e< \frac{s(r-k)(n-m)}{r-1+s(r-k)}$. By Remark \ref{rmk:2}, $\mathbf{HSE}(\bx)$ belongs to an $(s-1,\GL,m)_r$-periodic subspace. By Proposition \ref{prop:4.7}, we obtain a list of solutions of size $O(1/\Gz)$. Furthermore, by  \cite{Guru2012} the list can be computed in time ${\rm poly}(n,r^{\Gz})$.

The decoding radius of $\widehat{{\mC}}_q(n,k,m,r)$ is the same as the one of ${{\mC}}_q(n,k,m,r)$. By \eqref{eq:11}, we have
\[\widehat{{\tau}}=\tau\approx\frac c{c+1}\left(1-\frac{r-1}{r-c}\times R\right)\ge \frac c{c+1}\left(1-\frac{r-1}{r-c}\times \widehat{R}-\frac{r-1}{r-c}\times 2\Gz\right)\]
for $1\le c\le r-2$. Setting $\widehat{\Ge}=\frac{r-1}{r-c}\times 2\Gz$ gives the desired result.
\end{proof}
\begin{rmk}\label{rmk:6} In the code $\widehat{{\mC}}_q(n,k,m,r)$, if we set $s\approx 4/\Ge^2$,  $r\approx 4/\Ge^2$ and $k/(r-1)=\Ge/2$, then one gets the list decoding radius $\widehat{\tau}\approx 1-\widehat{R}-\widehat{\Ge}$. The list size is  $O(1/\zeta)=O(1/\widehat{\Ge})$. This proves Corollary \ref{cor:1.3}(ii).
\end{rmk}


\begin{thebibliography}{99}


\bibitem{Del}
P.~Delsarte,
\newblock {\it Bilinear forms over a finite field, with applications to coding theory,}
\newblock { J. Comb. Theory, Ser. A.},  {\bf 25}(1978), pp 226--241.

\bibitem{D15}	Y. Ding,
On List-Decodability of Random Rank Metric Codes and Subspace Codes. IEEE Transactions on Information Theory, {\bf 61(1)}(2015), 51--59.




\bibitem{Dvir2012}
Z.~Dvir and S.~Lovett, {\it
\newblock Subspace evasive sets,}
\newblock In {\em Proceedings of the 44th ACM Symposium on Theory of
  Computing}, pages 351--358, 2012.


\bibitem{ES09}
T.~Etzion and N.~Silberstein,
\newblock {\it Error-correcting codes in projective spaces via rank-metric codes and
  ferrers diagrams,}
\newblock { IEEE Transactions on Information Theory}, {\bf 55(7)}(2009),2909--2919.


\bibitem{faure}
C.~Faure,
\newblock {\it Average number of {G}abidulin codewords within a sphere,}
\newblock In {\em Int. Workshop on Alg. Combin. Coding Theory (ACCT)}, pages
  86--89, 2006.

\bibitem{FS}
M.~A. Forbes and A.~Shpilka,
\newblock {\it On identity testing of tensors, low-rank recovery and compressed
  sensing,}
\newblock In { Proceedings of the 44th ACM Symposium on Theory of
  Computing}, pages 163--172, 2012.

\bibitem{Ga85}
E.~M. Gabidulin.
\newblock Theory of codes with maximal rank distance.
\newblock {\em Problems of Information Transmission}, {\bf 21(7)}(1985), 1-12.

\bibitem{Ga91}
E.~M. Gabidulin,
\newblock {\it A fast matrix decoding algorithm for rank-error-correcting codes,} LNCS {\bf 573}, pages 126--133. Springer, 1991.

\bibitem{GPT91}
E.~M. Gabidulin, A.~V. Paramonov, and O.~V. Tretjakov,
\newblock {\it Ideals over a non-commutative ring and their applications in
  cryptology,} {\em EUROCRYPT}, LNCS {\bf 547}, pages 482--489, 1991.




\bibitem{Guru2013B}
V.~Guruswami and S.~Kopparty, {\it
\newblock Explicit subspace designs,}
\newblock In {\em Proceedings of the 54th IEEE Symposium on Foundations of
  Computer Science}, 2013.

\bibitem{GNW12}
V.~Guruswami, S.~Narayanan, and C.~Wang.
\newblock List decoding subspace codes from insertions and deletions.
\newblock In {\em Proceedings of the 3rd Innovations in Theoretical Computer
  Science Conference}, pages 183--189, January 2012.


\bibitem{GR08} V. Guruswami and A. Rudra, {\it Explicit codes achieving list decoding capacity: error correction with optimal redundancy,}  IEEE Trans. on Inform. Theory, {\bf 54(1)}(2008), 135--150.




\bibitem{GS99} V. Guruswami and M. Sudan, {\it Improved Decoding of Reed-Solomon and Algebraic-Geometric codes,}
IEEE Trans. on Inform. Theory, {\bf 45(3)}(1999), 1757-1767.






\bibitem{GW}
V.~Guruswami and C.~Wang.
\newblock Linear-algebraic list decoding for variants of {Reed-Solomon} codes.
\newblock {\em IEEE Transactions on Information Theory}, 59(6):3257--3268,
  2013.


  \bibitem{Guru2015}
V. Guruswami, C. Wang and C. Xing, {\it Explicit rank-metric and subspace codes list-decodable with optimal redundancy,} http://arxiv.org/abs/1311.7084, 2013.

\bibitem{Guru2012}
V.~Guruswami and C.~Xing,
\newblock {\it Folded codes from function field towers and improved optimal rate
  list decoding,}
\newblock {\em Electronic Colloquium on Computational Complexity (ECCC)},
  19:36, 2012.
\newblock Extended abstract appeared in the {\it Proceedings of the 44th ACM
  Symposium on Theory of Computing (STOC'12)}.

\bibitem{Guru2013}
V.~Guruswami and C.~Xing,
\newblock {\it List decoding {Reed-Solomon, Algebraic-Geometric, and Gabidulin}
  subcodes up to the {S}ingleton bound,}
\newblock {\em Electronic Colloquium on Computational Complexity (ECCC)},
  19:146, 2012.
\newblock Extended abstract appeared in the {\it Proceedings of the 45th ACM
  Symposium on Theory of Computing (STOC'13)}.




\bibitem{KK08}
R.~Koetter and F.~R. Kschischang,
{\it  Coding for errors and erasures in random network coding,}
\newblock { IEEE Transactions on Information Theory}, {\bf 54(8)}(2008), 3579--3591.


  \bibitem{MV12}
H.~Mahdavifar and A.~Vardy,
\newblock {\it List-decoding of subspace codes and rank-metric codes up to
  {S}ingleton bound,}  {\em CoRR}, abs/1202.0866, 2012.

\bibitem{RWZ15}
N. Raviv and A. Wachter-Zeh, {\it Some Gabidulin Codes cannot be List Decoded Efficiently at any Radius,} http://arxiv.org/abs/1501.04272, 2015.


\bibitem{S15} J. Sheekey, {\it
 A new family of linear maximum rank distance codes,} http://arxiv.org/pdf/1504.01581.pdf, 2015.

\bibitem{WZ13}
A. Wachter-Zeh,
{\it Bounds on List Decoding of Rank-Metric Codes,}
IEEE Transactions on Information Theory, {\bf 59(11)}(2013),  7268--7277.

\bibitem{WXS03}
{H. Wang and C. Xing and R. Safavi-Naini},
{\it Linear authentication codes: Bounds and constructions,}
 { IEEE Transactions on Information Theory},
{\bf {49}(4)}(2003), {866--872}.

\end{thebibliography}
\end{document}